\documentclass[12pt,draftclsnofoot,onecolumn]{IEEEtran}

\usepackage{amsfonts}
\usepackage{times}
\usepackage{latexsym}
\usepackage{amssymb}
\usepackage{amsmath}
\usepackage{cite}
\usepackage{verbatim}

\newcommand{\bydef}{\triangleq}

\def\SNR{{\textsf{SNR}}}

\def\bydef{:=}

\def\bb0{{\mathbb{0}}}

\def\bydef{:=}

\def\bb{{\mathbf{b}}}

\def\bh{{\mathbf{h}}}

\def\bs{{\mathbf{s}}}

\def\b0{{\mathbf{0}}}

\def\bB{{\mathbf{B}}}

\def\bI{{\mathbf{I}}}

\def\bbE{{\mathbb{E}}}

\def\bbN{{\mathbb{N}}}

\def\bbR{{\mathbb{R}}}

\def\bydef{:=}

\def\sf0{{\mathsf{0}}}

\def\nn{\nonumber}

\usepackage{graphicx}
\usepackage{amssymb}
\usepackage{amsfonts}
\usepackage{amsmath}
\usepackage{latexsym}
\begin{document}

\newtheorem{thm}{Theorem}
\newtheorem{lemma}{Lemma}
\newtheorem{rem}{Remark}
\newtheorem{exm}{Example}
\newtheorem{prop}{Proposition}
\newtheorem{defn}{Definition}
\newtheorem{cor}{Corollary}
\def\proof{\noindent\hspace{0em}{\itshape Proof: }}
\def\endproof{\hspace*{\fill}~\QED\par\endtrivlist\unskip}
\def\bh{{\mathbf{h}}}
\def\SNR{{\mathsf{SNR}}}
\title{Super Critical and Sub Critical Regimes of Percolation with Secure Communication}
\author{
Rahul~Vaze
\thanks{Rahul~Vaze is with the School of Technology and Computer Science, Tata Institute of Fundamental Research, Homi Bhabha Road, Mumbai 400005, vaze@tcs.tifr.res.in. }}

\date{}
\maketitle
\noindent
\begin{abstract}
Percolation in an information-theoretically secure graph is considered where both the legitimate and the eavesdropper nodes 
are distributed as Poisson point processes. For both the path-loss and the path-loss plus fading model, 
upper and lower bounds on the minimum density of the legitimate nodes 
(as a function of the density of the eavesdropper nodes) required for non-zero probability of having an 
unbounded cluster are derived. The lower bound is universal in nature, i.e.  the constant does not depend on the density of the 
eavesdropper nodes. 
\end{abstract}

\section{Introduction} 
Percolation theory studies the phenomenon of formation of unbounded connected clusters in large graphs \cite{BookRoy}, and 
percolation is defined as the event that there exists an unbounded connected cluster in a graph. 
Any wireless network can be naturally thought of as a graph, where the presence of an edge/connection between any two nodes 
can be defined in variety of ways \cite{BookRoy, BookPenrose}. Percolation in a wireless network corresponds to having long range connectivity, 
i.e. nodes that are far apart in space have a connected path between them. 
Thus, percolation theory is a natural tool to study the long-range connectivity in multi-hop wireless networks. 

Assuming the location of nodes of the wireless network to be distributed as a Poisson point process (PPP) with intensity $\lambda$, percolation has been studied 
for the Boolean model \cite{Gilbert1961}, where two nodes are connected if the distance between them is less than a fixed radius, for the SINR model \cite{Dousse2006}, where two nodes are connected if the SINR between them is greater than a threshold, and for the random connection model \cite{Penrose1991}, where two nodes are connected  with some probability that depends on the distance between them independently of other nodes. 
In all these works \cite{Gilbert1961,Dousse2006,Penrose1991}, a phase transition behavior has been established, i.e. there exists a critical intensity $\lambda_c$, where if 
$\lambda < \lambda_c$, then the probability of percolation is zero, while if $\lambda > \lambda_c$ then percolation happens almost surely. 

Recently, to study the existence of unbounded connected clusters in wireless networks in the presence of eavesdroppers, the concept of secure percolation model has been defined in \cite{Pinto2010,HaenggiSec2010}. The secure percolation model allows legitimate connected nodes to exchange information at a non-zero rate while maintaining perfect secrecy from all the eavesdroppers \cite{Wyner1975}.
For the path-loss model of signal propagation, where an edge  between a legitimate node $i$ and legitimate node $j$ exists, if node $j$ is closer to node $i$ than its nearest eavesdropper,  existence of the phase transition phenomenon has been established for  secure percolation in \cite{Pinto2010} assuming that the locations of the legitimate as well as the eavesdropper nodes are distributed as independent PPPs. 

In this paper, for studying secure percolation in wireless networks, we consider both the path-loss and path-loss plus fading model of signal propagation, where the
path-loss model is as described above, while in the path-loss plus fading model, two legitimate nodes $i$ and  $j$ are connected, if the signal power (product of path-loss and fading channel magnitude) received at node $j$ from node $i$ is larger than the signal power received at any other eavesdropper.  

The contributions of this paper are as follows.
We first derive an universal lower  bound on the critical intensity for secure percolation with the path-loss model, and later extend it to the path-loss plus fading model when the fading channel magnitudes have  finite support. We show that the critical intensity $\lambda_c \ge \frac{\lambda_E}{c}$, where $\lambda_E$ is the intensity of the eavesdropper process, and $c>1$ is a constant that does not depend on $\lambda_E$. Previously,  a lower bound has been obtained on $\lambda_c$ for the path-loss model in \cite{Pinto2010}, where the constant depends on $\lambda_E$. 
To show that $\lambda_c \ge \frac{\lambda_E}{c}$, we use the result of \cite{Gourre2008} on the critical intensity of percolation in random Boolean model. 
We also obtain upper bounds on the critical intensity for both the path-loss and path-loss plus fading models using a different approach compared to \cite{Pinto2010}, since the upper bound of \cite{Pinto2010} is not valid for the path-loss plus fading model.

{\it Notation:}
The expectation of function $f(x)$ with respect to $x$ is denoted by
${\bbE}(f(x))$. The modulus of  $x$ is denoted by $|x|$.
A circularly symmetric complex Gaussian random
variable $x$ with zero mean and variance $\sigma^2$ is denoted as $x
\sim {\cal CN}(0,\sigma^2)$.  $(x)^+$ denotes the function $\max\{x,0\}$. The cardinality of set $S$ is denoted by $\#(S)$. The complement of set $S$ is denoted by $S^c$. $S_2\backslash S_1$ represents the elements of $S_2$ that are not in its subset $S_1$. We denote the 
origin by ${\bf 0}$. 
A ball of radius $r$ centered at $x$ is denoted by $\bB(x,r)$. The boundary of a geometric object $G$ is denoted by $\delta G$. Area of region $B\in \bbR^2$ is denoted by $\nu(B)$. $H(x)$ denotes the entropy of random variable $x$ \cite{Cover2004}. 
 We use the symbol
$\bydef$  to define a variable.
\section{System Model}
\label{sec:sys}
Consider a wireless network consisting of the set of legitimate users denoted by $\Phi$, and the set of eavesdroppers  denoted by $\Phi_E$. 
We consider the secrecy graph model of \cite{HaenggiSec2010, Pinto2010} which is as follows. Let $x_i$ and $x_j$, $x_i,x_j\in\Phi$, want to communicate secretly, i.e. without providing any knowledge of their communication to any node in $\Phi_E$. 
Then to send a message $m$, $x_i$ sends a signal $\bs = (s(1),\dots,s(n))$ to $x_j$ over $n$ time slots. The received signals at $x_j$ ($rx_{j}$), and 
$e\in \Phi_E$ ($rx_{e}$), are 
\begin{equation}
rx_j(\ell) = d_{ij}^{-\alpha/2}h_{ij}s(\ell) + v_{ij}(\ell), \ \ell=1,2,\dots,n,
\end{equation}
and
\begin{equation}
rx_e(\ell) = d_{ie}^{-\alpha/2}h_{ie}s(\ell) + v_{ie}(\ell), \ \ell=1,2,\dots,n,
\end{equation}
respectively, 
where $d_{ij}$ and $d_{ie}$ are the distances between $x_i$ and $x_j$, and $x_i$ and $e$, respectively, $\alpha>2$ is the path loss exponent, $h_{ij}$ and $h_{ie}$ are the fading channel coefficients between $x_i$ and $x_j$, and $x_i$ and $e$, respectively, 
that is constant for $n$ time uses,  and 
$v_{ij}(\ell), v_{ie}(\ell) \sim {\cal CN}(0,1)$. We assume that $\bs, h_{ij}, h_{ie}, v_{ij}(\ell), v_{ie}(\ell)$ are independent of each other. Assuming an average power constraint of $P$ at each node in $\Phi$, i.e. 
$\frac{\sum_{\ell=1}^n \bbE\{|s(\ell)|^2\}}{n} \le P$,  the maximum rate of reliable communication between 
$x_i$ and $x_j$ such that an eavesdropper $e$ gets no knowledge about message $m$, i.e. $H(m| rx_e(1) \dots rx_e(n)) = H(m)$, is \cite{Wyner1975}
\[R_{ij}(e) \bydef \left[\log_2\left(1+ P d_{ij}^{-\alpha}|h_{ij}|^2\right) - \log_2\left(1+ P d_{ie}^{-\alpha}|h_{ie}|^2\right)\right]^+.\] 
Thus, $R_{ij}(e)$ is the communication rate between $x_i$ and $x_j$ that is secure from eavesdropper $e$. 
To consider communication between $x_i$ and $x_j$ that is secured from  all the eavesdropper nodes of $\Phi_E$, 
we define $R_{ij}$ as the rate of secure communication (secrecy rate) between $x_i$ and $x_j$ if 
\[R_{ij} \bydef \min_{e\in \Phi_E} R_{ij}(e).\] 


\begin{defn} Secrecy Graph \cite{Pinto2010}: Secrecy graph is a directed graph $SG(\gamma) \bydef \{\Phi, {\cal E}\}$, with vertex set $\Phi$, and edge 
set ${\cal E} \bydef \{(x_i, x_j): R_{ij}\ge \gamma\}$, where $\gamma$ is the minimum rate of secure communication required between any two nodes of $\Phi$. 
\end{defn}

\begin{defn} We define that there is a {\it path} from node $x_i$ to $x_j$ if there is a connected path from $x_i$ to $x_j$ in the SG.  A path between $x_i$ and $x_j$ on $SG(\gamma)$ is represented as $x_i \rightarrow x_j$.
\end{defn}

\begin{defn} We define that a node $x_i$ can {\it connect} to  $x_j$ if there is an edge between $x_i$ and $x_j$ in the $SG(\gamma)$.\end{defn}

Similar to \cite{Pinto2010, HaenggiSec2010}, in this paper we assume that the locations of $\Phi$ and $\Phi_E$ are distributed as independent homogenous 
Poisson point processes (PPPs) with 
intensities $\lambda$ and $\lambda_E$, respectively. 
The secrecy graph  when $\Phi$ and $\Phi_E$ are distributed as PPPs is referred to as the Poisson secrecy graph (PSG).
Moreover, we consider $\gamma =0$ in the rest of the paper, and drop the index $\gamma$ from the definition of PSG. Therefore there exists an edge between $x_i$ and $x_j$ in PSG if it can support a non-zero secrecy rate, $R_{ij} > 0$. Generalization to $\gamma >0$ is straightforward. We define the connected component of any node $x_j \in \Phi$, as $C_{x_j} \bydef   \{x_k \in \Phi, x_j\rightarrow x_k\}$, with cardinality  $|C_{x_j}|$. Note that because of stationarity of the PPP, the distribution of $|C_{x_j}|$ does not depend on $j$, and hence without loss of generality from here on we consider node $x_1$ for the purposes of defining connected components.

In this paper we are interested in studying the percolation properties of the PSG. In particular, 
we are interested in finding the minimum value of $\lambda$, $\lambda_c$, for which the probability of having an unbounded connected component in PSG is greater than zero as a function of  $\lambda_E$, i.e. $\lambda_c \bydef \inf \{{\lambda}:P(|{\cal C}_{x_1}| = \infty)>0\}$. 
The event  $\{|{\cal C}_{x_1}| = \infty\}$ is also referred to as {\it percolation} on PSG, and we say that percolation happens if 
$P(\{|{\cal C}_{x_1}| = \infty\})>0$, and does not happen if $P(\{|{\cal C}_{x_1}| = \infty\})=0$. From the Kolmogorov's zero-one law \cite{BookDurrett},  in a PPP percolation model, a phase transition behavior is observed, where below the critical density $\lambda < \lambda_c$ (subcritical regime), the probability of formation of unbounded connected components is zero, while for $\lambda > \lambda_c$ (supercritical regime) there is an unbounded connected component with probability one  \cite{BookRoy}.


\begin{rem} Note that we have defined PSG to be  a directed graph, and the component of $x_1$ is its out-component, i.e. the set of nodes with which $x_1$ can communicate secretly. 
Since $x_i \rightarrow x_j, \ x_i,x_j\in \Phi$, does not imply $x_j \rightarrow x_i \ x_i,x_j\in \Phi$,  one can similarly define in-component 
$C_{x_j}^{in} \bydef   \{x_k \in \Phi, x_k\rightarrow x_j\}$,  bi-directional component $C_{x_j}^{bd} \bydef   \{x_k \in \Phi, x_k\rightarrow x_j\ \text{and} \ x_k\rightarrow x_j\}$,  and either one-directional component $C_{x_j}^{ed} \bydef   \{x_k \in \Phi, x_k\rightarrow x_j \ \text{or} \ x_k\rightarrow x_j \}$. Percolation results for $C_{x_j}^{in}$, $C_{x_j}^{bd}$ and $C_{x_j}^{ed}$ follow similar to the results presented in this paper for $C_{x_j}$. 
\end{rem}

 
  \section{Path-Loss Model} 
  \label{sec:pl}
  With the path-loss model, where $h_{ij}=1, h_{ie}=1,$ for 
$\forall \ x_i,x_j\in \Phi, e\in \Phi_E$, \[R_{ij} \bydef \left[\log_2\left(1+ P d_{ij}^{-\alpha}\right) - \log_2\left(1+ P \max_{e\in \Phi_E}d_{ie}^{-\alpha}\right)\right]^+.\] With $\gamma =0$, 
  $PSG = \{\Phi, {\cal E}\}$, where the edge set ${\cal E} = \{(x_i, x_j): d_{ij}\le \min_{e\in \Phi_E}d_{ie}\}$, i.e. $x_i$ can connect to $x_j$,
  if  $x_j$ is closer to $x_i$ than any other eavesdropper.  Therefore, with $\gamma =0$, 
  in the path-loss model, node $x_i\in \Phi$ can connect to those nodes of $\Phi$ that are closer than its nearest 
  eavesdropper of $\Phi_E$. The maximum radius of connectivity of any node $x_i$ is denoted by $\rho(x_i) \bydef \min_{e\in \Phi_E}d_{ie}$. Because of the stationarity of the PPP, $\rho(x_i)$ is identically distributed for all $x_i$, and for simplicity we define $\rho$ to be 
  random variable which is identically distributed to $\rho(x_i)$ with probability density function (PDF) $\phi_{\rho}$.
  It is easy to show that $\bbE\{\rho^2\}= \frac{1}{\pi\lambda_E}$ \cite{Haenggi2005}. For the path-loss model, next, we discuss the sub-critical regime, and then follow it up with the super-critical regime.
  
\subsection{Sub-Critical Regime}
\label{sec:plsubc} In this section we are interested in obtaining a lower bound on $\lambda$ as a function of $\lambda_E$ for which the 
probability of percolation is zero.
  
Let  $D_m$ be a square box with side $2m$ centered at origin, i.e. $D_{m} = [-m \ m]\times [-m \ m]$. 
For $r>0$, consider any node $x_1\in \Phi\cap D_r$,\footnote{Without loss of generality we can assume that  $x_1$ is located at the origin.} and let ${\cal C}_{x_1}$ be its connected component.
Let  $x_L \in {\cal C}_{x_1}$ be the farthest node from $x_1$ in terms of Euclidean distance as shown in  Fig. \ref{fig:conn}. Let $r$  be chosen such that $x_L \in D_{10r}^c$. 

Let $A_{\cal B}(r)$ be the event that the maximum radius of connectivity of any node $x\in \Phi \cap {\cal B}$ is less than $r$, i.e. $A_{{\cal B}}(r) = \{\rho(x) \le r, \  \forall \ x  \in \Phi\cap {\cal B}\}$.
Let $B(q,r), \ q\in \bbR^2$, be the event that there is a path from a  node $x\in \Phi\cap q+D_r$ to a node  $y \in \Phi \cap q+D_{9r}\backslash q+D_{8r}$ with all the nodes on the path  between $x$ and $y$  lying inside $D_{10r}+q$, and the length of any edge of the path between  $x$ and $y$ is less than $r$. Note that due to stationarity $P(B(q,r)) = P(B({\bf 0},r))$.

In addition to the farthest node $x_L$ of ${\cal C}_{x_1}$ lying in $D_{10r}^c$, if $A_{D_{10r}}(r)$ also occurs, then there is a path from $x_1 \in D_r$ to node $y \in D_{9r}\backslash D_{8r}$  with all the nodes on the path  between $x$ and $y$  lying inside $D_{10r}$, since there is path between $x_1$ and $x_L$, and $\rho(x) \le r,   \forall \  x  \in \Phi\cap D_{10r}$, and the length of any edge of the path between  $x$ and $y$ is less than $r$. Therefore if $x_L \in D_{10r}^c$, and $A_{D_{10r}}(r)$ occurs, then $B({\bf 0},r)$ occurs. Hence the following proposition follows.

\begin{figure*}
\centering
\includegraphics[width=4.5in]{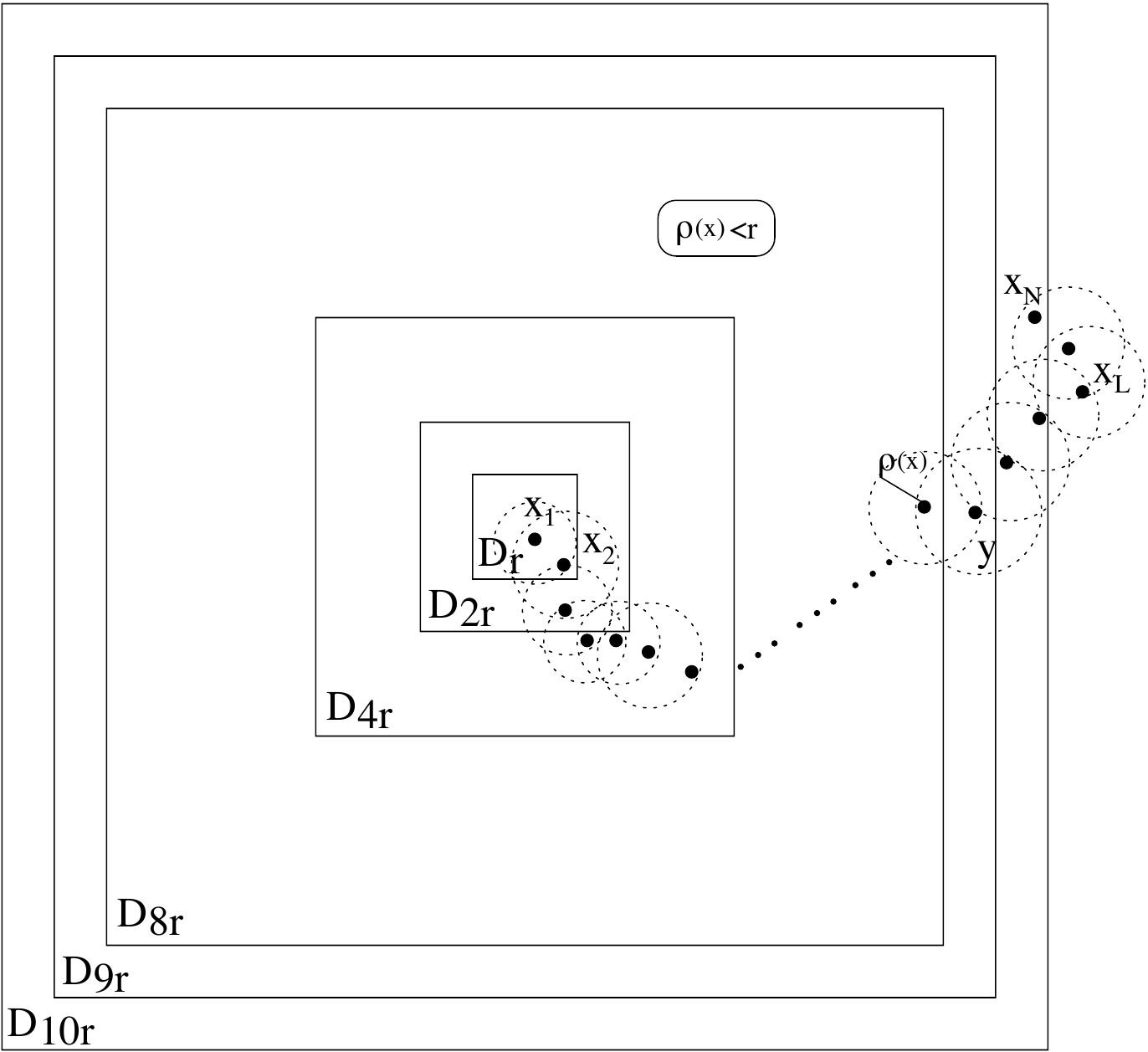}
\caption{Transmission capacity of the secondary network with multiple transmit and receive antennas.}
\label{fig:conn}
\end{figure*}

\begin{prop}\label{prop:transform}
\begin{equation}
P(x_L \in D_{10r}^c) \le P(B({\bf 0},r)) + P(A_{D_{10r}}(r)^c).
\end{equation}
\end{prop}

Note that $P(| {\cal C}_{x_1}|=\infty) \le \lim_{r \rightarrow \infty}P(x_L \in D_{10r}^c)$,  since 
infinitely many nodes of a PPP cannot lie in a finite region. 
It  easily follows that $P(A_{D_{10r}}(r)^c)\rightarrow 0$ as $r \rightarrow \infty$ (Proposition \ref{prop:a}).  Hence to 
show that $P(| {\cal C}_{x_1}|=\infty) = 0$ for $\lambda < \lambda_c$, it is sufficient to show that $P(B({\bf 0},r))$ goes to zero as $r \rightarrow \infty$ for $\lambda < \lambda_c$.

The main Theorem of this subsection is as follows.

\begin{thm}\label{thm:mainsubc}For $\lambda \le \frac{\pi\lambda_E}{4C^2}$,  where $C$ is a constant, $P(|{\cal C}_{x_1}| = \infty) =0$.
\end{thm}
\begin{proof} From Proposition \ref{prop:transform}, $ P(|{\cal C}_{x_1}| = \infty) \le \lim_{r \rightarrow \infty}P(B({\bf 0},r)) + P(A_{D_{10r}}(r)^c)$.  From Proposition \ref{prop:a} we get $\lim_{r \rightarrow \infty} P(A_{D_{10r}}(r)^c) = 0$, and from Lemma \ref{lem:conv}, for $\lambda \le \frac{1}{4C^2\bbE\{\rho^2\}} = \frac{\pi\lambda_E}{4C^2}$, since $\bbE\{\rho^2\} = \frac{1}{\pi\lambda_E}$,  $\lim_{r \rightarrow \infty} P(B({\bf 0},r)) = 0$.
\end{proof}

In the rest of the section, we prove  Proposition \ref{prop:a} and Lemma \ref{lem:conv} using ideas similar to \cite{Gourre2008}, where a lower bound on the critical density is derived for a random Boolean model. 
In a random Boolean model, nodes are spatially distributed as a PPP, and balls with i.i.d. radius are centered at each node of the PPP. The quantity of interest is the region spanned by the union of balls (also called the connected component). 
Secure percolation with the path-loss model is similar to the random Boolean model, since a legitimate node $x \in\Phi$ can connect to any other legitimate node within a radius $\rho(x)$ (radius of connectivity) that is determined by the nearest eavesdropper node. With secure percolation, however, the radii of connectivity of different legitimate nodes are not independent, and hence the proof of \cite{Gourre2008} does not apply directly. 


Next, we prove some intermediate results that are required for  proving Proposition \ref{prop:a} and Lemma \ref{lem:conv}.

\begin{prop}\label{prop:b} $P(B({\bf 0},r)) \le \lambda C_1 r^2$, where $C_1$ is a constant that only depends on the dimension of 
the PPP which in our case is two.
\end{prop}
\begin{proof} See Appendix \ref{app:propb}.\end{proof}

\begin{prop} \label{prop:a} For any $m\in \bbN$, $P(A_{D_{mr}}(r)^c)  \le \lambda C_2^2 \int_{r}^{\infty} s^2 \phi_{\rho}(s) ds$, where $C_2$ is a constant that only depends on $m$ and the dimension of 
the PPP, and  $P(A_{D_{mr}}(r)^c) \rightarrow 0$ as $r\rightarrow \infty$.
\end{prop}
\begin{proof}See Appendix \ref{app:propa}
\end{proof}
\begin{lemma}\label{lemma:independence} Event $B(q,r)$ only depends on $x \in \Phi\cap q+D_{9r}$, and $e \in \Phi_E\cap q+D_{10r}$.
\end{lemma}
\begin{proof} By definition, $B(q,r)$ is  the event that there is a path from a  node $x\in \Phi\cap q+D_r$ to a node  $y \in \Phi \cap q+D_{9r}\backslash q+D_{8r}$ with all the nodes on the path  between $x$ and $y$  lying inside $D_{10r}+q$, and the length of any edge of the path between  $x$ and $y$ is less than $r$. Thus, clearly, $B(q,r)$ only depends on $x \in \Phi\cap q+D_{9r}$. Moreover, since length of each edge of the path between  $x$ and $y$ is less than $r$, the event that a point $x \in \Phi\cap q+D_{9r}$ has an edge to $y\in \Phi\cap q+D_{9r}$ only depends on
$e \in \Phi_E\cap \bB(x,r)$. 
In the worst case, $x$ can be arbitrarily close to the boundary of $q+D_{9r}$, hence the event $B(q,r)$ only depends on $e \in \Phi_E\cap q+D_{10r}$.
\end{proof}


\begin{lemma}\label{lemma:main} $P(B({\bf 0},10r)) \le C_3 P(B({{\bf 0},r}))^2 + P(A_{D_{100r}}(r)^c)$, where $C_3$ is a constant that only depends on the dimension of the PPP. 
\end{lemma}
\begin{proof} See Appendix \ref{app:lemmain}.
\end{proof}

Recall that if we can show that $P(B({\bf 0},10r)) \rightarrow 0$ as $r \rightarrow \infty$, $P(|{\cal C}_{x_1}| = \infty) =0 $ follows. To show that 
 $P(B({\bf 0},10r)) \rightarrow 0$ as $r \rightarrow \infty$, we need the following result from \cite{Gourre2008}.

\begin{prop}\label{prop:goure} 
Let $f$ and $g$ be two measurable, bounded and non-negative functions from $[1, +\infty]$ to $\bbR^+$. If $f(x) \le 1/2$ for $x \in [1,10]$, and $g(x) \le 1/4$ for $x \in [1, +\infty]$, and $f(x) \le f(x/10)^2 + g(x)$ for $x\ge 10$, then $f(x)$ converges to $0$ as $x \rightarrow \infty$ whenever $g(x)$ converges to $0$ as $x \rightarrow \infty$.  
\end{prop}
\begin{proof}See Lemma $3.7$ \cite{Gourre2008}.
\end{proof}

Let $M= \frac{(\bbE\{\rho^2\})^{1/2}}{10}$, $f(r) \bydef C P(B({\bf 0}, Mr))$, and $g(r) \bydef \lambda C^2 \int_{\frac{Mr}{10}}^{\infty} s^2 \phi_\rho(s)ds$, where $C = \max\{C_1, C_2,C_3\}$. 
Then the following is true.
\begin{prop}\label{prop:sizef} For $\lambda \le \frac{1}{4C^2\bbE\{\rho^2\}}$, $f(r) \le \frac{1}{2}$ for $r\in [1,10]$.
\end{prop}
\begin{proof}
From Proposition \ref{prop:b}, $ C P(B({\bf 0}, Mr)) \le \lambda C^2 M^2 r ^2$, which using the definition of $M$ is $\le  \lambda C^2 \bbE\{\rho^2\} \left(\frac{r}{10}\right)^2$, from which the result follows.
\end{proof}

\begin{prop}\label{prop:sizeg} For $\lambda \le \frac{1}{4C^2\bbE\{\rho^2\}}$, $g(r) \le \frac{1}{4}, \ \forall \ r \in [1, +\infty]$.
\end{prop}
\begin{proof}
Note that $\int_{\frac{Mr}{10}}^{\infty} s^2 \phi_\rho(s)ds \le \bbE\{\rho^2\}$, hence  $g(r) \le \frac{1}{4}, \ \forall \ r$, when $\lambda \le \frac{1}{4C\bbE\{\rho^2\}}$.
\end{proof}

\begin{lemma}\label{lem:conv} For $\lambda \le \frac{1}{4C^2\bbE\{\rho^2\}}$, $f(r) \rightarrow 0$ and $P(B({{\bf 0},r})) \rightarrow 0$ as $r \rightarrow \infty$. 
\end{lemma}
\begin{proof} From Lemma \ref{lemma:main},  $f(r) \le f(r/10)^2+ g(r)$, and from Proposition \ref{prop:sizef} and \ref{prop:sizeg}, $f(r) \le \frac{1}{2}$ for $r \in [1,10]$ and $g(r) \le \frac{1}{4} \ \forall \ r$. Hence using Proposition \ref{prop:goure} it follows that $f(r)  \rightarrow 0$ and consequently $P(B({{\bf 0},r})) \rightarrow 0$ as $r \rightarrow \infty$. 
\end{proof}

{\bf Discussion:} In this section we obtained an universal lower bound on the critical intensity $\lambda_c$ required for percolation with  the path-loss model. 
Our proof is an adaptation of \cite{Gourre2008}, for the non-independent radii of connectivity. Note that a lower bound on $\lambda_c$ has been obtained in \cite{Pinto2010} for the path-loss model, however, our lower bound is universal, i.e. the constant in our lower bound does not change with the choice of $\lambda_E$ as was the case in \cite{Pinto2010}.
The 
main idea behind the proof is that if $\lambda$ is below a threshold (the derived lower bound), the probability that there is a path between two legitimate nodes at a distance $r$ from each other  goes to zero as $r \rightarrow \infty$. Therefore with probability one,  if $\lambda$ is below a threshold, the connected component of any node lies inside a bounded region, and since infinitely many nodes of a PPP do not lie in a bounded region, the connected component of any node is finite. 

\subsection{Super-Critical Regime}
\label{sec:plsupc}
In this section we derive an upper bound on $\lambda_c$ for the path-loss model using a different approach compared to \cite{Pinto2010}. Our upper bounding technique is applicable for both the path-loss as well as the path-loss plus fading model, while the upper bound derived in \cite{Pinto2010} is valid only for the path-loss model. Before deriving the upper bound, we briefly discuss the approach of \cite{Pinto2010}.
The upper bound on $\lambda_c$ for the path-loss model has been derived in \cite{Pinto2010} by coupling the continuum percolation on the PPP to 
the discrete lattice percolation. The corresponding discrete lattice is a lattice on $\bbR^2$ with edge length $\psi$, where an edge is defined to be open if there is at least one node of $\Phi$ inside each square on either side of the edge and there is no node of 
$\Phi_E$ in an union of circles of radius (a function of $\psi$) around the edge. The analysis in  \cite{Pinto2010} critically depended on the fact that the two legitimate nodes can connect if the distance between them is less compared to their nearest eavesdropper. Since with the path-loss plus fading model,  two legitimate nodes can connect even if the distance between them is mode compared to their nearest eavesdropper, the upper bound obtained in \cite{Pinto2010} does not apply to the path-loss plus fading model.

Our upper bound on $\lambda_c$ for the path-loss model is summarized in the next Theorem.
\begin{thm} \label{thm:plsupc} For the path-loss model, $ \exists \ \epsilon \in (0,1), N_1\in \bbN$ for which 
$P(|{\cal C}_{x_1}| = \infty) >0$ if  
 $\lambda > \frac{\lambda_E}{1-(1-\epsilon)e^{-\lambda_E\pi N_1^2}}$.
\end{thm} 
\begin{proof} We prove the Theorem by contradiction.
Define a ball ${\bf B}({\bf 0},n), \ n \in \bbN$  to be {\it open} if all nodes $x\in \Phi \cap {\bf B}({\bf 0},n)$ can connect to at least one node in $x\in \Phi \cap {\bf B}({\bf 0},n)^c$, and {\it closed} otherwise. Let there be no percolation,  i.e. $ P(|{\cal C}_{x_1}| = \infty)=0$ for any $x_1\in \Phi$, 
then $\exists \ N_0 \in \bbN$ such  that ${\bf B}({\bf 0},N_0)$ is closed, since otherwise there will be percolation. Therefore, if $ P(|{\cal C}_{x_1}| = \infty)=0$, then $P(\cup_{n\in\bbN}{\bf B}({\bf 0},n) \ \text {is closed} ) = 1$, and $\sum_{n=1}^{\infty} P({\bf B}({\bf 0},n) \ \text {is closed} ) \ge 1$. Therefore, $\exists \ \epsilon \in (0,1), N_1\in \bbN$ such that $P({\bf B}({\bf 0},N_1)  \text{ is closed}) \ge \epsilon$. 
Note that the event that ${\bf B}({\bf 0},N_1)$ is closed implies that there is at least one node of $x\in \Phi \cap {\bf B}({\bf 0},N_1)$ that cannot connect to any node of $x\in \Phi \cap {\bf B}({\bf 0},N_1)^c$. Therefore,  
\begin{eqnarray*}
P({\bf B}({\bf 0},N_1) \  \text {is closed} ) &\le& 
P(x  \in \Phi \cap {\bf B}({\bf 0},N_1) \  \text {is not connected to any node in} \ \Phi \cap  {\bf B}({\bf 0},N_1)^c ),\\
&=&   P( \min_{e\in \Phi_E}d_{xe} < \min_{y \in \Phi \cap {\bf B}({\bf 0},N_1)^c}d_{xy} ), \ \ x  \in \Phi \cap {\bf B}({\bf 0},N_1).
\end{eqnarray*}
Moreover, note that  it is easiest  for a node $x\in \Phi \cap {\bf B}({\bf 0},N_1)$ to be not able to connect to a node $y \in \Phi \cap {\bf B}({\bf 0},N_1)^c$, if $x$ is at the origin. Therefore, we have that for $ x  \in \Phi \cap {\bf B}({\bf 0},N_1)$
\[P( \min_{e\in \Phi_E}d_{xe} < \min_{y \in \Phi \cap {\bf B}({\bf 0},N_1)^c}d_{xy} ) \le 
P(\min_{e\in \Phi_E}d_{{\bf 0} e} < \min_{y \in \Phi \cap {\bf B}({\bf 0},N_1)^c}d_{{\bf 0} y} ),\] where $D_e \bydef \min_{e\in \Phi_E}d_{{\bf 0} e}$ is the distance of the nearest eavesdropper from the origin, and $D_l(N_1) \bydef \min_{y \in \Phi \cap {\bf B}({\bf 0},N_1)^c}d_{{\bf 0} y}$ is the distance of the nearest legitimate node belonging to $\Phi \cap  {\bf B}({\bf 0},N_1)^c$ from the origin. 
From Proposition \ref{prop:upboundFE}, we have that  $P(D_e< D_l(N_1)) = 1-e^{-\lambda_E\pi N_1^2}\frac{\lambda}{\lambda+\lambda_E}$. Therefore, if $ P(|{\cal C}_{x_1}| = \infty)=0$, then $1-e^{-\lambda_E\pi N_1^2}\frac{\lambda}{\lambda+\lambda_E} > \epsilon$, and 
$\lambda < \frac{\lambda_E}{1-(1-\epsilon)e^{-\lambda_E\pi N_1^2}}$.
\end{proof}

\begin{prop}\label{prop:upboundFE}   $P(D_e< D_l(N_1)) = 1-e^{-\lambda_E\pi N_1^2}\frac{\lambda}{\lambda+\lambda_E}$. \end{prop}
 \begin{proof}
 Let $D_l(N_1) = N_1 + X$, where $X$ is the random variable representing the shortest distance between the node $y^\star$
 , $y^{\star} =   \arg  \min_{y \in \Phi \cap {\bf B}({\bf 0},N_1)^c}d_{{\bf 0} y}$, and the disc of radius $N_1$. 
Note that $P(X>x)$ is equivalent to $|(\Phi\cap ({\bf B}(0,N_1+x) \backslash {\bf B}({\bf 0},N_1))| =0$. Thus,  $P(X>x) = e^{-\pi \lambda ((x+N_1)^2-N_1^2)}$. Differentiating, we get the PDF $f_X(x) = \lambda \pi 2(x+N_1)e^{-\pi \lambda (x^2+2xN_1)}$. Thus,
 \begin{eqnarray*}
P(D_e< D_l(N_1)) &=&  \bbE_{D_l(N_1)}\{1-e^{-\lambda_E\pi y^2}\}, \\
&=&  \bbE_{X}\{1-e^{-\lambda_E\pi (x+N_1)^2}\}, \\
&=&  1-2 \lambda \pi \int_{0}^{\infty} e^{-\lambda_E\pi (x^2+N_1^2+2xN_1)} (x+N_1)e^{-\pi \lambda (x^2+2xN_1)} dx,  \\
&=&  1-2 \lambda \pi e^{-\lambda_E\pi N_1^2}\int_{0}^{\infty} e^{-\lambda_E\pi (x^2+2xN_1)} (x+N_1)e^{-\pi \lambda (x^2+2xN_1)} dx,  \\
&=& 1-e^{-\lambda_E\pi N_1^2}\frac{\lambda}{\lambda+\lambda_E}.
\end{eqnarray*}
\end{proof}

{\bf Discussion:} In this section we derived an upper bound on the critical intensity $\lambda_c$ required for percolation with the path-loss model. To obtain an upper bound that is valid for the path-loss as well as the path-loss plus fading model, we take a different approach 
compared to \cite{Pinto2010}. We define a ball with radius ($n\in \bbN$) centered at the origin to be open if all the legitimate nodes lying inside the ball are able to connect to at least one node lying outside the ball. Therefore, if there is no percolation, then at least one of the balls is closed, and  there exists an $\epsilon \in (0,1)$ and $N_1$ for which the probability of the ball with radius $N_1$ is closed is greater than $\epsilon$. Since the probability of the ball with radius $N_1$ to be closed is upper bounded by the probability that a node lying at origin is unable to connect to any node outside of a ball of radius $N_1$, the required upper bound is  derived by finding the probability that a node lying at origin is unable to connect to any node outside of a ball of radius $N_1$.

\section{Path-loss plus Fading Model}\label{sec:fading}
In this section we consider signal propagation in the presence of fading in addition to the path-loss. While considering fading together with 
path-loss with $ \gamma =0$,  $PSG = \{\Phi, {\cal E}\}$, with vertex set $\Phi$, and edge 
set ${\cal E} = \{(x_i, x_j): R_{ij}\ge 0 \}$, where
 \[R_{ij} \bydef \left[\log_2\left(1+ P d_{ij}^{-\alpha}|h_{ij}|^2\right) - \log_2\left(1+ P \max_{e\in \Phi_E}d_{ie}^{-\alpha}|h_{ie}|^2\right)\right]^+.\]
 Therefore there exists an edge between $x_i$ and $x_j$ if $ d_{ij}^{-\alpha}|h_{ij}|^2 >  \max_{e\in \Phi_E}d_{ie}^{-\alpha}|h_{ie}|^2$.
Next, we discuss the sub-critical regime, 
  and then follow it up with the super-critical regime.

\subsection{Sub-critical Regime}\label{sec:fadingsubc}
We assume that all the the channel coefficient magnitudes $|h_{ij}|^2, |h_{ie}|^2, x_i,x_j\in \Phi, e\in \Phi_E$ are bounded above, i.e. $\exists \ \kappa \in \bbN$ such that $|h_{ij}|^2 \le \kappa, |h_{ie}|^2\le \kappa, x_i,x_j\in \Phi, e\in \Phi_E$. Essentially, what we need is that the channel coefficient magnitudes should not have infinite support. Most often in literature, channel coefficient magnitudes are assumed to be exponentially distributed (derived from Rayleigh fading distribution), however, in practice, it is not difficult to safely assume that the channel coefficient magnitudes are upper bounded by some large constant. 

With the bounded channel magnitude assumption, we will essentially reuse the proof we developed in Section \ref{sec:plsubc} for the sub-critical regime for the path-loss model as follows. Let $x_1\in \Phi$, and $\max_{e\in \Phi_E}d_{1e}^{-\alpha}|h_{1e}|^2  > \beta$, i.e. the maximum of the received 
power at 
any eavesdropper from $x_1$ is greater than $\beta$. Then a necessary condition for $x_1$ to connect to $x_j$ is that $d_{1j}^{-\alpha}|h_{1j}|^2 > \beta$. Then using our assumption that $|h_{ij}|^2 \le \kappa$, we know that 
$d_{1j}^{-\alpha}|h_{1j}|^2 \le d_{1j}^{-\alpha}\kappa$, and hence $d_{1j}^{-\alpha}\kappa > \beta$. Thus, $x_1$ can possibly communicate secretly with only those  $x_k's \in \Phi$ that are at a distance 
less than $\eta \bydef \left(\frac{\kappa}{\beta}\right)^{\frac{1}{\alpha}}$ from it. To draw a parallel with the setup of Section  \ref{sec:plsubc} for the sub-critical regime, this is equivalent to assuming that the radius of connectivity of $x_1$ is less than $r$, $\rho(x_1) < r$, and $\eta$ is going to play the role of $r$.

Let $G_{D_{10\eta}}(\beta) \bydef  
\{ \max_{e\in \Phi_E}d_{je}^{-\alpha}|h_{je}|^2  > \beta, \forall \ x_j \in \Phi\cap D_{10\eta}\}$, i.e. $G_{D_{10\eta}}(\beta)$ is the event 
that the maximum received power at any eavesdropper from all nodes of $\Phi \cap D_{10\eta}$ is greater than $\beta$. Therefore, in turn 
this guarantees that any node of $\Phi$ that lies inside  $ D_{10\eta}$ can only connect to nodes of $\Phi$ which are a distance of $\eta$ or less. Event $G_{D_{10\delta}}(\beta)$ is equivalent to event $A_{D_{10r}}(r)$ of Section  \ref{sec:plsubc}. Moreover, let $Q(\eta)$ be the event that there is a path from node $x\in\Phi\cap D_{\eta}$ to a node $y\in\Phi\cap D_{9\eta} \backslash D_{8\eta}$ with all the nodes of the path  between $x$ and $y$ 
lying inside $D_{10\eta}$, and the distance between any two nodes on the path  between $x$ and $y$ is less than $\eta$.  Event $Q(\eta)$ is equivalent to event $B(q,r)$ from Section  \ref{sec:plsubc} with $q={\bf 0}$. Let $x_1\in \Phi\cap D_{\eta}$, and  the connected 
component of $x_1$ be ${\cal C}_{x_1}$. Let the farthest node of ${\cal C}_{x_1}$ be $x_L$, and $x_L \in D_{10\eta}^c$. From Proposition \ref{prop:transform}, it follows that 
\[P(x_L \in D_{10\eta}^c) \le P(Q(\eta)) + P(G_{D_{10\eta}}(\beta)^c).\]
Note that $\beta \rightarrow 0$ is same as $\eta \rightarrow \infty$ which is equivalent to $r \rightarrow \infty$ from Section  \ref{sec:plsubc}.
Similar to Proposition \ref{prop:a}, we can show that $P(G_{D_{10\eta}}(\beta)^c) \rightarrow 0$, as $\beta \rightarrow 0$. 
Moreover, notice that now the problem is identical to the problem while considering only path-loss (Section  \ref{sec:plsubc}), with 
$\eta$ playing the role of $r$. Thus, following the proof of Theorem \ref{thm:mainsubc}, we obtain  the following Theorem.
\begin{thm}\label{thm:mainfadingsubc} For the path-loss plus fading model, if $\lambda \le \frac{\lambda_E}{4C^2}$, where $C>0$ is a constant that does not depend on $\lambda$ or $\lambda_E$, then $P(|{\cal C}_{x_1}| = \infty) =0$ if the 
channel coefficient magnitudes have finite support.
\end{thm}
 {\bf Discussion:} 
 In this section we obtained an universal lower bound on the critical intensity $\lambda_c$ required for percolation for the path-loss plus fading model. We reused the proof derived in Section \ref{sec:plsubc} for the path-loss model, by assuming that all the fading channel coefficients have a bounded support. The bounded support assumption on channel magnitudes allows us to  conclude that if the maximum signal power received at any eavesdropper is above a threshold $\beta$, then each legitimate node 
 can connect to any other legitimate node only if it is at bounded distance (a function of $\beta$) from it. Therefore, with this assertion, we show that if $\lambda$ is below a threshold (the derived lower bound), the probability that there is a path between two legitimate nodes at a distance $r$ from each other  goes to zero as $r \rightarrow \infty$. Therefore with probability one,  if $\lambda$ is below a threshold, the connected component of any node lies inside a bounded region, and since infinitely many nodes of a PPP do not lie in a bounded region, the connected component of any node is finite.

\subsection{Super-critical Regime}\label{sec:fadingsupc}
In this section we obtain an upper bound on $\lambda_c$ for the path-loss plus fading model. We assume that the fading channel 
coefficients $h_{ij}, h_{ie}, \ \forall \ i,j \in \Phi, e \in \Phi_E$ are distributed as ${\cal CN}(0,1)$, to model a rich scattering wireless environment. 
Note that the results derived in this section can 
be generalized for any distribution of  the fading channel coefficients.
Similar to the previous subsection, in this subsection also, we will reuse the proof we developed in Section \ref{sec:plsupc} for the super-critical regime for the path-loss model as follows. Previously, in \cite{Pinto2010}, an upper bound on the critical intensity for the path-loss model 
is obtained by mapping the continuum percolation model to a  discrete percolation model depending on distance between the nodes. 
The strategy used in \cite{Pinto2010}, however, cannot be extended to the path-loss plus fading model since in this case 
$x_i\rightarrow x_j, x_i,x_j \in \Phi$ even if 
$d_{i,e} < d_{ij}$ since it is possible to have $|h_{ie}|^2d^\alpha_{i,e} < |h_{ij}|^2d^{-\alpha}_{ij}$ when $d_{i,e} < d_{ij}$.

\begin{thm}\label{thm:plsupc} For the path-loss plus fading model,  $\exists \ \epsilon \in (0,1),  N_1\in \bbN$ for which  
$P(|{\cal C}_{x_1}| = \infty) >0$  if  
 $\lambda > \frac{\lambda_E \nu (1-\epsilon)}{\epsilon \nu_1}$, where 
 $\nu_1 = \pi \int_{N_1}^{\infty} x^{2/\alpha}e^{-x}dx$, and $\nu = \pi \int_{0}^{\infty} x^{2/\alpha}e^{-x}dx$. 
\end{thm} 
\begin{proof}
Let there be no percolation, i.e. $ P(|{\cal C}_{x_1}| = \infty)=0$ for any $x_1\in \Phi$. 
Assume that $x_1$ lies at the origin. Similar to Section  \ref{sec:plsupc}, define 
a ball ${\bf B}({\bf 0},n)$  to be {\it open} if all nodes of $x\in \Phi \cap {\bf B}({\bf 0},n)$ can connect to at least one node in $x\in \Phi \cap {\bf B}({\bf 0},n)^c$, and {\it closed} otherwise.
Then with no percolation, $\sum_{n=1}^{\infty} P({\bf B}({\bf 0},n) \ \text {is closed} ) \ge 1$. 
Thus,  $\exists \ \epsilon \in (0,1),  N_1$ such that $P({\bf B}({\bf 0},N_1) \ \text {is closed}) \ge \epsilon$.
Note that 
\begin{eqnarray*}
P({\bf B}(0,1) \  \text {is closed} ) & \le & 
P(x  \in \Phi \cap {\bf B}({\bf 0},N_1) \  \text {is not connected to any node in} \ \Phi \cap  {\bf B}({\bf 0},N_1)^c ), \\
&\le& P(\max_{e\in \Phi_E}|d_{x e}|^{-\alpha}|h_{xe}|^2 > \max_{y \in \Phi \cap {\bf B}({\bf 0},N_1)^c}d_{xy}^{-\alpha}|h_{x y}|^2 ) , \ \ x  \in \Phi \cap {\bf B}({\bf 0},N_1).
\end{eqnarray*}
Moreover, note that  it is most difficult for a node $x\in \Phi \cap {\bf B}({\bf 0},N_1)$ to connect to a node $y \in \Phi \cap {\bf B}({\bf 0},N_1)^c$ if $x$ is at the origin. 
Therefore, we have that for $x\in \Phi \cap {\bf B}({\bf 0},N_1)$
\[P(\max_{e\in \Phi_E}|d_{x e}|^{-\alpha}|h_{xe}|^2 > \max_{y \in \Phi \cap {\bf B}({\bf 0},N_1)^c}d_{xy}^{-\alpha}|h_{x y}|^2 ) \le 
P(\max_{e\in \Phi_E}|d_{{\bf 0} e}|^{-\alpha}|h_{{\bf 0} e}|^2 > \max_{y \in \Phi \cap {\bf B}({\bf 0},N_1)^c}d_{{\bf 0} y}^{-\alpha}|h_{{\bf 0} y}|^2 ),\] where $\Delta \bydef \max_{e\in \Phi_E}|d_{{\bf 0} e}|^{-\alpha}|h_{{\bf 0} e}|^2$ is the maximum of the power received by any eavesdropper from  the origin, and  $\Gamma  \bydef \max_{y \in \Phi \cap {\bf B}({\bf 0},N_1)^c}d_{{\bf 0} y}^{-\alpha}|h_{{\bf 0} y}|^2 $ is the maximum power received by any legitimate node belonging to $\Phi \cap  {\bf B}(0,1)^c$ from the origin. Note that $\Delta$, and $\Gamma$ correspond to $D_e$, and $D_{l}(N_1)$, respectively, from the proof of Theorem \ref{thm:plsupc}. 
 From Proposition \ref{prop:powercomp}, we obtain that $P(\Delta > \Gamma) = \frac{\lambda_E\nu}{\lambda_E\nu + \lambda\nu_1}$, where 
 $\nu_1 = \pi \int_{N_1}^{\infty} x^{2/\alpha}e^{-x}dx$, and $\nu = \pi \int_{0}^{\infty} x^{2/\alpha}e^{-x}dx$. 
 Recall that $P({\bf B}({\bf 0},N_1) \ \text {is closed} )> \epsilon $ which implies that $P(\Delta > \Gamma)> \epsilon $, and hence 
$\frac{\lambda_E\nu}{\lambda_E\nu + \lambda\nu_1}  > \epsilon$. Thus, if $P(|{\cal C}_{x_1}| = \infty) =0$, then $\lambda < \frac{\lambda_E \nu (1-\epsilon)}{\epsilon \nu_1}$, and therefore completes the proof.
\end{proof}

\begin{prop}$\label{prop:powercomp} P(\Delta > \Gamma) = \frac{\lambda_E\nu}{\lambda_E\nu + \lambda\nu_1}$, where $\nu_1 = \pi \int_{1}^{\infty} x^{2/\alpha}e^{-x}dx$, and $\nu = \pi \int_{0}^{\infty} x^{2/\alpha}e^{-x}dx$.
\end{prop}
\begin{proof} See Appendix \ref{app:powercomp}.
\end{proof}

\appendices
\section{Proof of Proposition \ref{prop:b}}
\label{app:propb}
Event $B({\bf 0},r)$ implies that  $\#(\Phi \cap D_{10r}) >0$. Hence $P(B({\bf 0},r)) \le P(\#(\Phi \cap D_{10r}) >0)$. Since $\bbE\{\#(\Phi \cap D_{10r})\}= \lambda \nu(D_{10r}) r^2$ is clearly greater than or equal to  $P(\#(\Phi \cap D_{10r}) >0)$, the result follows.

\section{Proof of Proposition \ref{prop:a}}
\label{app:propa}
Note that $1- P(A_{D_{mr}}(r)) = P(\exists \text{ at least one node} \  x \in  \Phi \cap D_{mr}, \  \text {such that} \  \rho(x) > r)$.  Hence 
\begin{eqnarray}\nonumber
1-P(A_{D_{mr}}(r)) & = & \sum_{j=0}^\infty P(\#(\Phi \cap D_{mr})=j) P(  \{\rho(x_1) > r\} \cup \ldots \cup \{\rho(x_j) > r\} ),\\\nn
&\stackrel{(a)}\le & \sum_{j=0}^\infty P(\#(\Phi \cap D_{mr})=j)  jP(  \{\rho(x) > r\}),\\ \nn
 &=  & \sum_{j=0}^\infty \frac{(\lambda \nu(D_{mr}))^j}{j!}e^{-\lambda \nu(D_{mr})}jP(  \{\rho(x) > r\}), \\\nn
 & = & \lambda \nu(D_{m}) r^2P(\{\rho(x) > r\}), \\ \nn
& \le & \lambda \nu(D_{m}) \int_{r}^{\infty} s^2 \phi_{\rho}(s) ds, \\ \nn
&\le & \lambda \nu(D_{m})
 \bbE\{\rho^2 \bI_{\{\rho>r\}}\},
 \end{eqnarray} 
where $(a)$ is obtained by using the union bound and since $\{\rho(x_j) \le r\}$ is identically distributed $\forall \ j$. Since $\bbE\{\rho^2\}$ is finite, 
$P(A_{D_{mr}}(r)^c) \rightarrow 0$ as $r\rightarrow \infty$.

\section{Proof of Lemma \ref{lemma:main}}
\label{app:lemmain}
Let $K$ and $L$ be two finite subsets of $\bbR^2$, such that $K \subset \delta D_{10}, L\subset \delta D_{80}$, and $D_{10}\backslash D_{9} \subset K+D_1, D_{81}\backslash D_{80}\subset L+D_1$. For example, see Fig. \ref{fig:covering} where black dots represent the points of $K\subset \delta D_{10}$ covering  $D_{10}\backslash D_{9}$ using $D_1$.  Let $C_3$ be the product of the cardinality of $K$ and $L$. 
Assume that $B({\bf 0},10r)$ and $A_{D_{100r}}(r)$ occur. Thus there exists a node in $\zeta \in D_{10r}\backslash D_{9r}$ which is connected to a node in $D_{90r}\backslash D_{80r}$. By the definition of $K$, $\zeta \in rk+ D_r$ for some $k\in K$. See Fig. \ref{fig:proof} for a pictorial description. Moreover since 
$A_{D_{100r}}(r)$ also occurs, node $\zeta$ is connected to some node in $\zeta+D_{9r}\backslash D_{8r}$ with all nodes lying inside $D_{10r}+\zeta$, and length of each edge is less than $r$. Hence if $B({\bf 0},10r)$ and $A_{D_{100r}}(r)$ occur, then $\cup_{k\in K} B(rk,r)$ happens, where $P(B(rk,r)) = P(B({{\bf 0},r}))$ for any $k\in K$. Similarly looking at nodes around 
$D_{80r}$ and using the definition of $L$ we can show that if $B({\bf 0},10r)$ and $A_{D_{100r}}(r)$ occur then $\cup_{\ell \in L} B(r\ell,r)$ happens. Hence if both $B({\bf 0},10r)$ and $A_{D_{100r}}(r)$ occur simultaneously, then $\cup_{k\in K} B(rk,r) \cap \cup_{\ell \in L} B(r\ell,r)$ happens, where $P(B(r\ell,r)) = P(B({{\bf 0},r}))$ for any $\ell \in L$. 
From Lemma \ref{lemma:independence}, we know that  
the event $\cup_{k\in K} B(rk,r)$ depends only upon the nodes of $\Phi$ and $\Phi_E$ lying in $D_{20r}$, while the event $\cup_{\ell \in L} B(r\ell,r)$ depends only upon the nodes of $\Phi$ and $\Phi_E$ lying in $D_{69r}^c$. Since $D_{20r}$ and $D_{69r}^c$ are disjoint, and since $\Phi$ are $\Phi_E$ are independent PPPs,  the events  $\cup_{\ell \in L} B(r\ell,r)$ and $\cup_{k \in K} B(rk,r)$ are independent, and hence we get that
$P(B({\bf 0},10r)\cap A_{D_{100r}}(r)) \le C_3 P(B({{\bf 0},r}))^2$.

\section{Proof of Proposition \ref{prop:powercomp}}
\label{app:powercomp}
Let the signal power received from $x_1\in \Phi$ (located at origin) at the $n^{th}$ node of $\Phi$ be $I_n \bydef 
d_{1n}^{-\alpha}|h_{1n}|^2$, and $e^{th}$ eavesdropper of $\Phi_E$ be $I_e^{E} \bydef d_{1e}^{-\alpha}|h_{1e}|^2$. 
Note that since $h_{1n}$ and $h_{1e}$ are Rayleigh distributed, the channel power gains
$|h_{1n}|^2$ and $|h_{1e}|^2$ are exponential distributed. Let the PDF of $|h_{1n}|^2$ be  $\chi_n(x)$, and $|h_{1e}|^2$ be $\chi_e(x)$.
Let $g_1, g_2 >0$,
then define two Marked Point Processes
\[{\cal P}(g_1) = \left\{(x_n, I_n)\ | \ x_n \in \Phi \cap \bB({\bf 0},N_1)^c, 
\ I_n > g_1, \right\},\]
and
\[{\cal P}_E(g_2) = \left\{(e_m, I_m^E)\ | \ e_m \in \Phi_E,
\ I_m^E > g_2 \right\}.\] 
Let the mean number of nodes in the set ${\cal P}(g_1)$
be $\lambda_{g_1}$, and set ${\cal P}_E(g_2)$
be $\lambda^E_{g_2}$. Since $|h_{1n}|^2$ and $|h_{1e}|^2$ are independent $\forall \ n ,e$,  by Marking Theorem \cite{Stoyan1995}, 
both ${\cal P}(g_1)$ and ${\cal P}_E(g_2)$ are Poisson point processes, and
\begin{eqnarray*}
\lambda_{g_1} & = & \lambda \int_{N_1}^{\infty}
\int_{0}^{\left(\frac{x}{g_1}\right)^{\frac{1}{\alpha}}} 2\pi r \chi_n(x) dr dx, \\
&=& \frac{\lambda\pi}{g_1^{2/\alpha}} \int_{N_1}^{\infty} x^{2/\alpha}e^{-x}dx, \\
&\bydef& \lambda g_1^{-\delta}\nu_1,
\end{eqnarray*}
and 
\begin{eqnarray*}
\lambda^E_{g_2} & = & \lambda_E \int_{0}^{\infty}
\int_{0}^{\left(\frac{x}{g_2}\right)^{\frac{1}{\alpha}}} 2\pi r \chi_e(x) dr dx, \\
&=& \frac{\lambda_E\pi}{g_2^{2/\alpha}} \int_{0}^{\infty} x^{2/\alpha}e^{-x}dx, \\
&\bydef& \lambda_E g_2^{-\delta}\nu,
\end{eqnarray*}
where $\mu = \frac{2}{\alpha}$, $\nu_1 = \pi \int_{N_1}^{\infty} x^{2/\alpha}e^{-x}dx$, and $\nu = \pi \int_{0}^{\infty} x^{2/\alpha}e^{-x}dx$.

Then the cumulative density function (CDF) of $\Gamma = \max_{x_n \in \Phi \cap \bB({\bf 0},N_1)^c} I_{n}$ is equal to the
probability that there are no nodes of $\Phi$ in the set ${\cal P}(g_1)$ \cite{Haenggi2005}.
Thus,
\begin{eqnarray*}
P\left(\Gamma\le g\right) =  e^{-\lambda \nu_1 g_1^{-\delta}}.
\end{eqnarray*}
Similarly, the CDF of the largest received power at any eavesdropper  $\Delta = \max_{e \in \Phi_E} I_{e}^E$ is equal to the
probability that there are no nodes of $\Phi$ in the set ${\cal P}_E(g_2)$
\begin{eqnarray*}
P\left(\Delta \le g_2\right) =  e^{-\lambda_E \nu g_2^{-\delta}}.
\end{eqnarray*}
Thus, 
\begin{eqnarray*}
P(\Delta > \Gamma) &=&  \int_{0}^{\infty}e^{-\lambda \nu_1 g^{-\delta}} dP(\Delta \le g), \\
&=&  \int_{0}^{\infty}e^{-\lambda \nu_1 g^{-\delta}} 
\delta \lambda_E \nu
g^{-\delta  -1}
e^{-\lambda_E\nu g^{-\delta}}\ dg,\\
&=& \frac{\lambda_E\nu}{\lambda_E\nu + \lambda\nu_1}.
\end{eqnarray*}

\bibliographystyle{IEEEtran}
\bibliography{IEEEabrv,Research}

\begin{figure}
\centering
\includegraphics[width=4.5in]{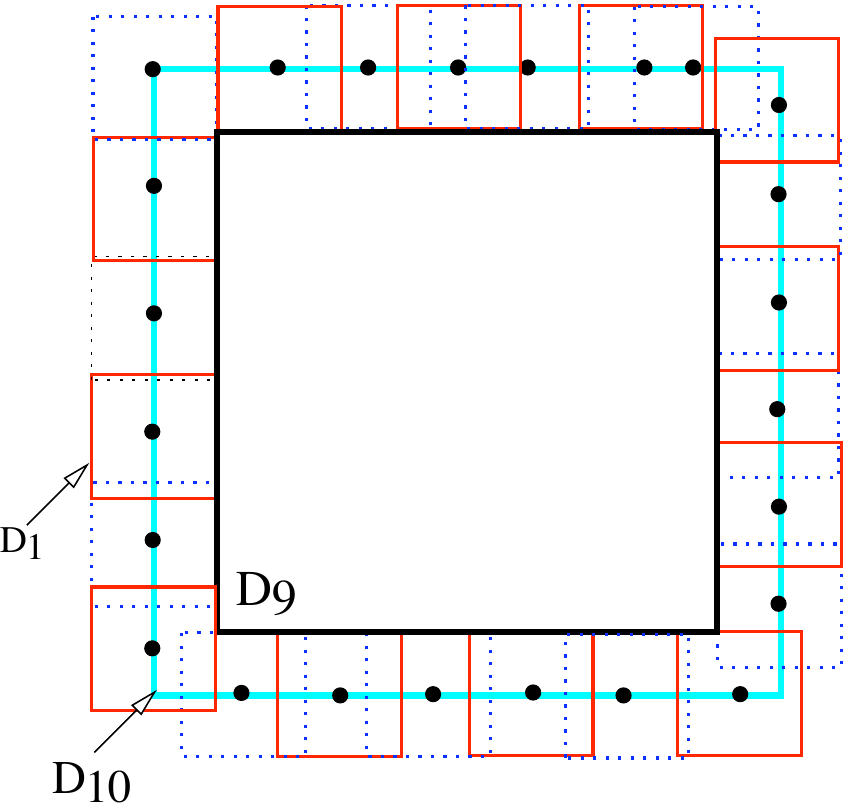}
\caption{Covering of $D_{10}\backslash D_{9}$ by discrete points lying on the boundary (black dots) of $D_{10}$ using $D_1$.}
\label{fig:covering}
\end{figure}

\begin{figure}
\centering
\includegraphics[width=4.5in]{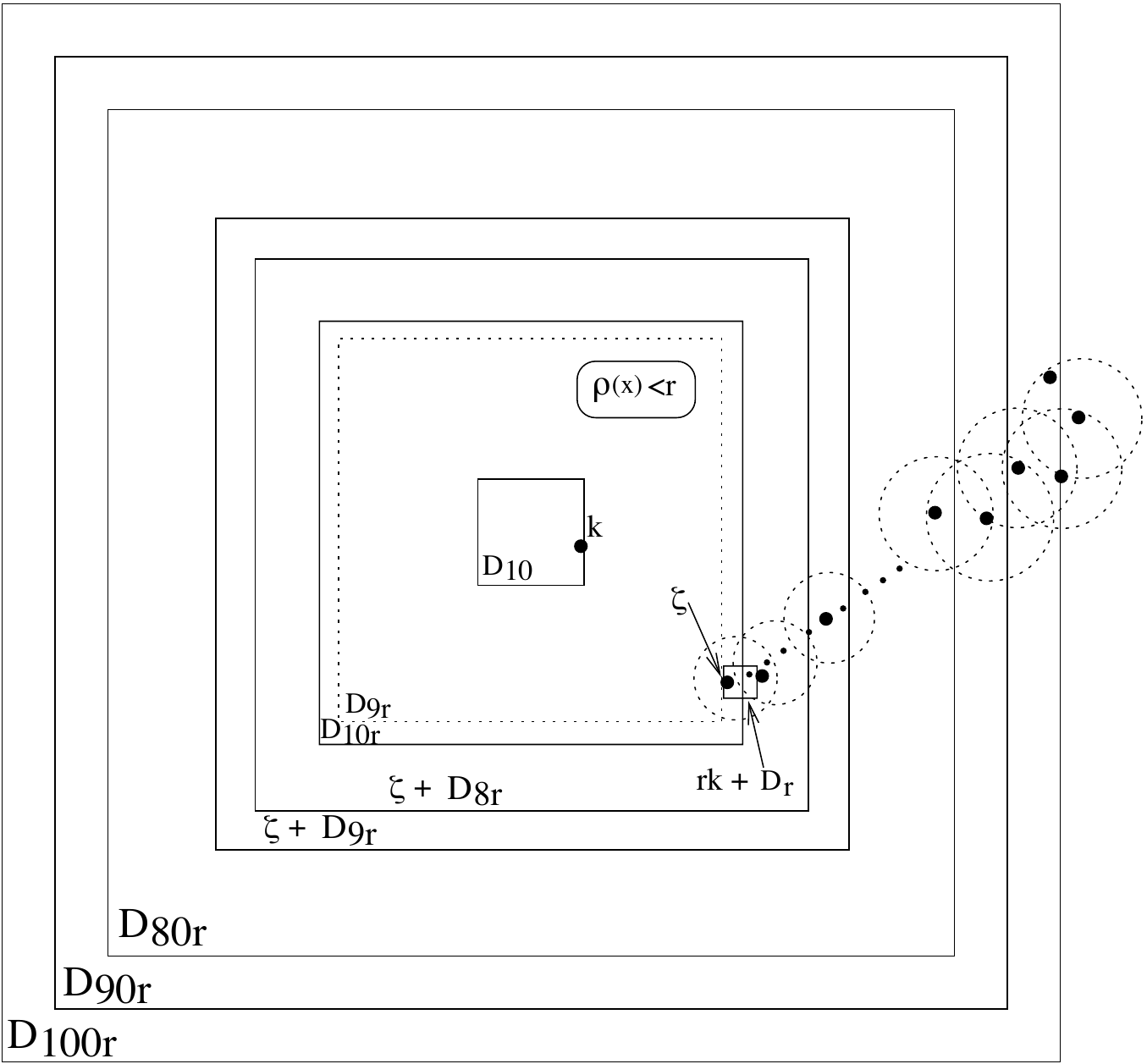}
\caption{Transmission capacity of the secondary network with multiple transmit and receive antennas.}
\label{fig:proof}
\end{figure}

\end{document}